\documentclass[letterpaper,USenglish]{lipics}

\usepackage{hyperref}
\usepackage{verbatim}
\usepackage{amsmath,amsfonts,amsmath,amssymb,amsthm}
\usepackage{latexsym}
\usepackage{yfonts}
\usepackage{graphicx}
\usepackage{bm}
\usepackage{bbm}
\usepackage{color}
\usepackage{xspace}
\usepackage{authblk}

\usepackage{microtype}

\theoremstyle{plain}
\newtheorem{prop}[theorem]{Proposition}

\newcommand{\ket}[1]{| #1 \rangle}
\newcommand{\bra}[1]{\langle #1 |}

\newcommand{\R}{\Bbb{R}}
\newcommand{\N}{\Bbb{N}}

\newcommand{\set}[1]{\{#1\}}
\newcommand{\Set}[2]{\set{#1 \,|\, #2}}
\newcommand{\x}{\times}
\newcommand{\eps}{\varepsilon}
\newcommand{\E}{\mathbb{E}}                    
\newcommand{\epsclose}[1][\eps]{\approx_{#1}}

\newcommand{\prob}[2][]{P#1[\,#2\, #1]}                 
\newcommand{\Prob}[3][]{P#1[\,#2\,#1|\, #3 \, #1]}      

\newcommand{\mb}[1]{\text{\boldmath$#1$}}     

\newcommand{\const}{c}
\newcommand{\av}[1]{\overline{#1}}             

\newcommand{\ev}{{\cal E}}                                   
\newcommand{\typ}{{\cal T}}                                   
\newcommand{\ind}[1]{^{#1}}
\newcommand{\indl}{\ind{\ell}}

\newcommand{\etal}{{\it et al.}\xspace}

\newcommand{\assign}{\ensuremath{\kern.5ex\raisebox{.1ex}{\mbox{\rm:}}\kern -.3em =}}

\renewcommand{\H}{{\cal H}} 
\newcommand{\regA}{{\sf A}}

\newcommand{\povm}{E} 
\newcommand{\POVM}{{\bf E}} 

\newcommand{\game}{{\cal G}} 

\newcommand{\q}{q}
\newcommand{\V}{{\sf V}}
\newcommand{\qn}{\q^{(n)}}
\newcommand{\val}{v}
\newcommand{\ns}{\text{\rm ns}}
\newcommand{\qu}{\text{\rm qu}}
\renewcommand{\c}{\text{\rm c}}
\newcommand{\pio}{\pi_\circ}

\newcommand{\A}{{\cal A}} 
 
\newcommand{\X}{{\cal X}}
\newcommand{\Y}{{\cal Y}}
\newcommand{\allX}{\mb{X}\hspace{-0.2ex}}
\newcommand{\allA}{\mb{A}}
\newcommand{\allx}{\mb{x}}
\newcommand{\alla}{\mb{a}}

\title{On the Parallel Repetition of Multi-Player Games: The No-Signaling Case}
\titlerunning{On the Parallel Repetition of Multi-Player Games} 

\author[1,2]{Harry Buhrman\thanks{h.buhrman@cwi.nl}}
\author[1]{Serge Fehr\thanks{s.fehr@cwi.nl}}
\author[2,1]{Christian Schaffner\thanks{c.schaffner@uva.nl}}
\affil[1]{Centrum Wiskunde \& Informatica (CWI), Amsterdam, The Netherlands}
\affil[2]{Institute for Logic, Language and Computation (ILLC), \hspace{5cm} University of Amsterdam, The Netherlands}

\authorrunning{H. Buhrman, S. Fehr, and C. Schaffner} 

\Copyright{Harry Buhrman, Serge Fehr, and Christian Schaffner}


\serieslogo{}
\volumeinfo
  {Billy Editor, Bill Editors}
  {2}
  {Conference title on which this volume is based on}
  {1}
  {1}
  {1}
\EventShortName{}
\DOI{10.4230/LIPIcs.xxx.yyy.p}

\begin{document}

\maketitle

\begin{abstract}
We consider the natural extension of two-player nonlocal games to an arbitrary number of players. An important question for such nonlocal games is their behavior under parallel repetition. For {\em two-player} nonlocal games, it is known that both the {\em classical} and the {\em non-signaling} value of any game converges to zero exponentially fast under parallel repetition, given that the game is non-trivial to start with (i.e., has classical/non-signaling value $<1$). Very recent results~\cite{DSV13arxiv,CS13arxiv,JPY13arxiv} show similar behavior of the {\em quantum} value of a two-player game under parallel repetition. For nonlocal games with three or more players, very little is known up to present on their behavior under parallel repetition; this is true for the classical, the quantum and the non-signaling value. 

In this work, we show a parallel repetition theorem for the {\em non-signaling} value of a large class of multi-player games, for an arbitrary number of players.
Our result applies to all multi-player games for which all possible combinations of questions have positive probability; this class
in particular includes all {\em free} games, in which the questions to the players are chosen independently. Specifically, we prove that if the original game $\game$ has a non-signaling value $v_\ns(\game) < 1$, then the non-signaling value of the $n$-fold parallel repetition is exponentially small in $n$. Stronger than that, we prove that the probability of winning more than $(v_\ns(\game) + \delta) \cdot n$ parallel repetitions is exponentially small in $n$ (for any~$\delta > 0$). 

Our parallel repetition theorem for multi-player games is weaker than the known parallel repetition results for two-player games in that the rate at which the non-signaling value of the game decreases not only depends on the non-signaling value of the original game (and the  number of possible responses), but on the complete description of the game. 
Nevertheless, we feel that our result is a first step towards a better understanding of the parallel repetition of nonlocal games with more than two players. 
\end{abstract}

\section{Introduction}\label{sec:intro}

\paragraph*{Background.}
In an $m$-player nonlocal game $\game$, $m$ players receive respective questions $x_1,\ldots,x_m$, chosen according to some joint probability distribution, and the task of the $m$ players is to provide ``good'' answers $a_1,\ldots,a_m$, {\em without communicating} with each other. The players are said to {\em win} the game if the given answers jointly satisfy some specific property with respect to the given questions. The {\em value} of a given game is defined to be the maximal winning probability of the players. One distinguishes between the classical, the quantum, and the non-signaling value, depending on whether the players are restricted to be classical, may share entanglement and do quantum measurements, or are allowed to make use of any hypothetical strategy that does not violate non-signaling.  

An important question for nonlocal games is their behavior under parallel repetition. This question is somewhat understood in the case of {\em two} players, where $m = 2$. Indeed, Raz showed in his celebrated parallel repetition theorem~\cite{Raz98} that if the classical value of a two-player game $\game$ is $v_\c(\game) < 1$ then the classical value $v_\c(\game^n)$ of the $n$-fold parallel repetition of $\game$ satisfies $v_\c(\game^n) \leq \bar{v}_\c(\game)^{n/\log(s)}$, where $s$ denotes the number of possible pairs of answers $a_1$ and $a_2$, and $\bar{v}_\c(\game) < 1$ only depends on $v_\c(\game)$. Raz's result was improved and simplified by Holenstein~\cite{Holenstein09}, who gave an explicit and tighter dependency between $\bar{v}_\c(\game)$ and $v_\c(\game)$, namely $\bar{v}_\c(\game) = 1 - \frac{1}{6000}(1-v_\c(\game))^3$. Holenstein also showed that a similar result holds for the non-signaling value of any two-player game: $v_\ns(\game^n) \leq \bar{v}_\ns(\game)^{n}$ for $\bar{v}_\ns(\game) = 1 - 
\frac{1}{6400}(1-
v_\ns(\game))^2$. 
Parallel repetition results for the quantum value of two-player games were first derived for certain special classes of games, like XOR-games~\cite{CSUU08} or unique games~\cite{KRT10}, or for a non-standard parallel repetition where the different repetitions of the original game are intertwined with modified versions of the original game~\cite{KV11}. Recently, several results about the parallel repetition of more general quantum games have been obtained~\cite{DSV13arxiv,CS13arxiv,JPY13arxiv}.

There are further improvements to the above results on two-player games. For instance, Rao~\cite{Rao11} showed a {\em concentration} result for the classical value of any two-player game, saying that the probability to win more than $(v_\ns(\game) + \delta) \cdot n$ out of the $n$ repetitions is exponentially small (for any $\delta > 0$).%
\footnote{Rao claims the concentration result only for the classical value, but the same techniques also apply to the non-signaling value. } 
Furthermore, he improved the bound on the classical value under parallel repetition for {\em projection} games. 
A similar improvement on the bound on the classical value under parallel repetition was given by Barak \etal~\cite{BRRRS09} for {\em free} games, together with a further improvement, namely a {\em strong} parallel repetition theorem (meaning that meaning that $v_\c(\game^n) \leq v_\c(\game)^{\Omega(n)}$), for {\em free projection} games. 

When considering multi-player nonlocal games with strictly more than $2$ players, to the best of our knowledge, very little is known about their behavior under parallel repetition, except for trivial cases. This applies to the classical, the quantum, and the non-signaling value.  
In~\cite{Rosen10}, Rosen proved a parallel-repetition result for more than 2 players. While her proof strategy is very similar to ours (closely following~\cite{Holenstein09}), a somewhat unnatural definition of multi-player non-signaling correlations is used where no $m-1$ provers together can signal to the remaining prover. In our (standard) model, one also demands that any subset (of arbitrary size) of provers can not signal to the remaining provers.

Another result about multi-player games is by Bri{\"e}t \etal \cite{BBLV13} about the related question of XOR repetition. They show the existence of a 3-player XOR game whose classical value of the XOR repetition is bounded from below by a constant (independent of the number of repetitions). Hence, XOR repetition does not hold for this game (but parallel repetition might still hold). Our result does not imply anything about those games, because the non-signaling value of XOR games is always 1.


Possible applications of our result could be of cryptographic nature where the hardness of a basic task is amplified by parallel repetition. A likely scenario for applying our results (and our original motivation to study the problem) is position-based quantum cryptography~\cite{BCFGGOS11,BFSS13}, in the spirit of a recent result on parallel repetition of a particular game~\cite{TFKW13}. However, as our result only applies to a restricted class of games, we were not able yet to apply it to in this cryptographic context.

\paragraph*{Our Results.}
We show a parallel repetition and a concentration theorem for the non-signaling value of $m$-player games for any $m$, for a large class of games. 
The class of games to which our result applies consists of all multi-player games with {\em complete support}, meaning that all possible combinations of questions $x_1,\ldots,x_m$ must have positive probability of being asked. This class of games in particular includes all {\em free} games, in which the questions to the different players are chosen independently.
For any $m$-player game $\game$ with complete support, we show that if $v_\ns(\game) < 1$ then there exists $\bar{v}_\ns(\game) < 1$ so that $v_\ns(\game^n) \leq \bar{v}_\ns(\game)^n$, and the probability of winning more than $(v_\ns(\game) + \delta) \cdot n$ out of the $n$ repetitions with an arbitrary non-signaling strategy is exponentially small (for any $\delta > 0$). 

We point out that our parallel repetition result for multi-player games (with complete support) is of a weaker nature than the parallel repetition results for two-player games discussed above, in that in our result the constant $\bar{v}_\ns(\game)$ depends on the complete description of the game $\game$, and not just on its non-signaling value $v_\ns(\game)$. 
Still, our result is the first that shows a parallel repetition result for a large class of $m$-player games with $m > 2$ for one of the three values (the classical, quantum or non-signaling) of interest.  

For proving our results, we borrow and extend tools from~\cite{Holenstein09} and~\cite{Rao11}, and combine them with some new technique. The new technique involves considering strategies that are {\em almost} non-signaling, meaning that the non-signaling properties only hold up to some small error. We then show (Proposition~\ref{prop:robust}) and use in our proof that the non-signaling value of a game is {\em robust} under extending the quantification over all non-signaling strategies to all almost non-signaling strategies.

\section{Preliminaries}

\subsection{Basic Notation}

For any $m$-partite set $\X =  \X_1\x\cdots\x\X_m$, any $m$-tuple $x = (x_1,\ldots,x_m) \in \X$, and any index set $I = \set{i_1,\ldots,i_k} \subseteq \set{1,\ldots,m}$, we write $\X_I$ to denote the  $k$-partite set $\X =  \X_{i_1}\x\cdots\x\X_{i_k}$, and we write $x_I$ to denote the $k$-tuple $x = (x_{i_1},\ldots,x_{i_k}) \in \X_I$. 
To denote elements from the $n$-fold Cartesian product of an $m$-partite set $\X$ as above, we write $\allx=(x\ind{1},\ldots,x\ind{n}) \in \X \x \cdots \x \X$ with $x\ind{i} = (x_1\ind{i},\ldots,x_m\ind{i}) \in \X$. 
For $i \in \set{1,\ldots,m}$, we then write $\allx_i$ for $\allx_i = (x_i\ind{1},\ldots,x_i\ind{n})$, and for $I = \set{i_1,\ldots,i_k} \subseteq \set{1,\ldots,m}$, $x_I\indl$ is naturally understood as $x_I\indl = (x_{i_1}\indl,\ldots,x_{i_k}\indl)$ and  $\allx_I$ as $\allx_I = (\allx_{i_1},\ldots,\allx_{i_k})$. Corresponding notation is used for random variables $X$ over $\X$ and $\allX$ over $\X \x \cdots \x \X$.

\subsection{Probabilities and Random Variables}

We consider finite probability spaces, given by a non-empty finite sample space $\Omega$ and a probability function $P:\Omega \to [0,1]$. 
A random variable is a function $X: \Omega \rightarrow \cal X$ from $\Omega$ into some finite set $\X$. The distribution of $X$, denoted as $P_X$, is given by $P_X(x) = \prob{X\!=\!x} = \prob{\Set{\omega \in \Omega}{X(\omega)\!=\!x}}$. 
The joint distribution of a pair of random variables $X$ and $Y$ is denoted by $P_{XY}$, i.e., $P_{XY}(x,y) = \prob{X\!=\!x \wedge Y\!=\!y}$, and the conditional distribution of $X$ given $Y$ is denoted by $P_{X|Y}$ and defined as $P_{X|Y}(x|y) = P_{XY}(x,y)/P_Y(y)$ for all $x$ and $y$ with $P_Y(y) > 0$. 
An event $\ev$ is a subset of $\Omega$, and the conditional distribution of a random variable $X$ given $\ev$ is denoted as $P_{X|\ev}$ and given by $P_{X|\ev}(x) = \prob{X \!=\!x \wedge \ev}/\prob{\ev}$.

The variational (or statistical) distance between two probability
distributions $P_X$ and $Q_X$ for the same random variable $X: \Omega \rightarrow \cal X$ over two probability spaces $(\Omega,P)$ and $(\Omega,Q)$ (with the same~$\Omega$),
is defined as
\begin{align*}
  \| P_X - Q_X \| \assign \frac12 \sum_{x \in \X} |P_X(x) - Q_X(x)| 
\end{align*}
If $P_X$ and $Q_X$ are $\eps$-close in variational distance, we also
write $P_X \epsclose Q_X$.

Usually, we leave the probability space $(\Omega,P)$ etc.\ implicit, and understand random variables $X,Y,\ldots$ to be defined by their joint distribution $P_{XY\cdots}$, or by some ``experiment'' that uniquely determines their joint distribution.

\subsection{Some Useful Facts}

The following lemma states that the variational distance cannot increase when less information is taken into account. 
\begin{lemma} \label{lemma:marginals}
Let $P_{XY}$ and $Q_{XY}$ be joint distributions for random variables $X$ and $Y$ with respective ranges $\X$ and $\Y$, and let $P_X$ and $Q_X$ be the corresponding marginals. Then,
$$\| P_X - Q_X \| \leq \|P_{XY} - Q_{XY}\| \, .$$
\end{lemma}
\begin{proof}
\begin{align*}
\| P_X - Q_X \| &= \frac12 \sum_{x \in \X} | P_X(x) - Q_X(x) | = \frac12 \sum_{x \in \X} \left| \sum_{y \in \Y} \big( P_{XY}(x,y) - Q_{XY}(x,y) \big)
\right| \\
&\leq \frac12 \sum_{x \in \X} \sum_{y \in \Y} \left| P_{XY}(x,y) - Q_{XY}(x,y)
\right| = \| P_{XY} - Q_{XY} \| \, .
\end{align*}
\end{proof}
The next lemma is due to Holenstein~\cite{Holenstein09} (a simplified version of his Corollary~6). 
\begin{lemma}\label{lemma:bound}
Let $T$ and $U\ind{1},\ldots,U\ind{L}$ be random variables with distribution $P_{TU\ind{1}\cdots U\ind{L}} = P_T \cdot P_{U\ind{1}|T}\cdots P_{U\ind{L}|T}$ (i.e. the $U\indl$'s are conditionally independent given $T$), and let $\ev$ be an event. Then
$$
\sum_{\ell=1}^L \,\bigl\|P_{TU\indl|\ev} - P_{T|\ev} \cdot P_{U\indl|T} \bigr\| \leq \sqrt{L \log\Bigl(\frac{1}{\prob{\ev}}\Bigr)} \, .
$$
\end{lemma}
The following is Hoeffding Inequality's for sampling without replacement~\cite{hoeffding63}.
\begin{theorem}[Hoeffding Inequality for sampling without replacement]
\label{thm:hoeffding}
Let $w \in \set{0,1}^n$ be an $n$-bit string with $\frac{1}{n}\sum_{\ell=1}^n w_i = \av{w}$. Let the random variables $D_1, D_2,\ldots, D_K$ be obtained by sampling $K$ random entries from $w$ {\em without replacement}. Then, for any $\eps > 0$, the random variable $\av{D}:=\frac{1}{K}\sum_k D_k$ satisfies
$$
\prob[\big]{\av{D} \leq \av{w} - \eps } \leq \exp\bigl(-2\eps^2 K\bigr) \, .
$$
\end{theorem}
%
Finally, we will make use of the Azuma-Hoeffding Inequality, stated below. We first define the notion of a supermartingale. 

\begin{definition}[Supermartingale]
A sequence of real valued random variables $M_0,M_1,\ldots,M_K$ is called a {\em supermartingale} if $\E[M_k|M_0\cdots M_{k-1}] \leq M_{k-1}$ (with probability $1$) for every $k \geq 1$.  
\end{definition}

\begin{theorem}[Azuma-Hoeffding Inequality]
If $M_0,M_1,\ldots,M_K$ is a supermartingale with $M_k \leq M_{k-1} + 1$, then
$$
\prob[\big]{M_K > M_0 + \eps K} \leq \exp\bigl(-\eps^2 K/2\bigr) \, .
$$
\end{theorem}

\subsection{Nonlocal Games}

\begin{definition}
An {\em $m$-player nonlocal game}, or simply {\em ($m$-player) game} $\game$ consists of two $m$-partite sets $\X =  \X_1\x\cdots\x\X_m$ and $\A =  \A_1\x\cdots\x\A_m$,  a probability distribution $\pi: \X \to [0,1]$ on $\X$, i.e., $\sum_x \pi(x) = 1$, and a verification predicate $\V: \X \x \A \to \set{0,1}$. 
\end{definition}

\begin{definition}
A {\em strategy} for an $m$-player game $\game = (\X,\A,\pi,\V)$  is a conditional probability distribution $\q(\cdot|\cdot):\A \x \X \to [0,1]$, i.e., $\sum_a \q(a|x) = 1$ for all $x \in \X$. 
\end{definition}

\begin{definition}\label{def:game}
For any  $m$-player game $\game = (\X,\A,\pi,\V)$ and any strategy $\q$ for $\game$, the {\em value} of the game with respect to $\q$ is given by 
$$
\val[\q](\game):= \sum_{x \in \X \atop a\in \A} \pi(x) \, \q(a|x) \, \V(x,a) \, .
$$
\end{definition}

Any $m$-player game $\game = (\X,\A,\pi,\V)$ and any strategy $\q$ for $\game$ together naturally define a probability space with random variables $X = (X_1,\ldots,X_m)$ and $A = (A_1,\ldots,A_m)$ with joint probability distribution $P_{XA}$ given by $P_{XA}(x,a) = \pi(x) \q(a|x)$.  The random variable $X$ describes the choice of the input $x \in \X$ according to $\pi$, and the random variable $A$ then describes the reply $a \in \A$ chosen according to the distribution $\q(\cdot|x)$. It obviously holds that $P_X = \pi$, and $P_{A|X}(\cdot|x) = \q(\cdot|x)$ for any $x \in \X$ with $P_X(x) > 0$. A subtlety is that for $x \in \X$ with $P_X(x) = 0$, the distribution $P_{A|X}(\cdot|x)$ is strictly speaking not defined whereas $\q(\cdot|x)$ is. 
The value of the game with respect to strategy $\q$ can be written in
terms of these random variables as $\val[\q](\game) = \prob{\V(X,A)
  \!=\! 1}$. In the following we define the classical, quantum and
non-signaling values of $m$-player games. Only the last one will be used in
the rest of the paper, but we provide all of them for the sake of completeness.

\begin{definition}\label{def:classical}
A {\em strategy} $\q$ for an $m$-player game $\game = (\X,\A,\pi,\V)$  is {\em classical} (or {\em local}) if there exists a probability distribution $p$ on a set $\cal W$ and conditional probability distributions $\q_1,\ldots,\q_m$ such that 
$$
\q(a_1,\ldots,a_m|x_1,\ldots,x_m) = \sum_{w \in \cal W} p(w) \prod_{i=1}^m \q_i(a_i|x_i, w) \, .
$$
The {\em classical value} of a game $\game$ is defined as $\val_\c(\game) \assign \sup_\q \val[\q](\game)$, where the supremum is over all classical strategies $\q$ for $\game$. 
\end{definition}

\begin{definition}\label{def:quantum}
A {\em strategy} $\q$ for an $m$-player game $\game = (\X,\A,\pi,\V)$
is {\em quantum} if there exists an $m$-partite quantum state
$\ket{\psi} \in \H_{\regA_1} \otimes \cdots \otimes \H_{\regA_m}$ and
for every $x = (x_1,\ldots,x_m) \in \X$ there exist POVMs $\POVM^1_{x_1} = \set{\povm^1_{x_1,a_1}}_{a_1 \in \A_1}, \ldots, \POVM^m_{x_m} = \set{\povm^m_{x_m,a_m}}_{a_m \in \A_m}$ such that for all $a = (a_1,\ldots,a_m) \in \A$ and $x = (x_1,\ldots,x_m) \in \X$:
$$
\q(a|x) = \bra{\psi} E^1_{x_1,a_1} \otimes\cdots \otimes E^m_{x_m,a_m} \ket{\psi}
$$
The {\em quantum value} of a game $\game$ is defined as $\val_\qu(\game) \assign \sup_\q \val[\q](\game)$, where the supremum is over all quantum strategies $\q$ for $\game$. 
\end{definition}

\begin{definition}\label{def:NS}
A {\em strategy} $\q$ for an $m$-player game $\game = (\X,\A,\pi,\V)$
is {\em non-signaling} if for any index subset $I \subset
\set{1,\ldots,m}$ and its complement $J = \set{1,\ldots,m} \setminus
I$, it holds that
$$
\sum_{a_J \in \A_J} \!\! \q(a_I,a_J|x_I,x_J) = \!\! \sum_{a_J \in \A_J} \!\! \q(a_I,a_J|x_I,x'_J) 
\quad\text{for all $a_I \in \A_I$, $x_I \in \X_I$ and $x_J,x'_J \in \X_J$} \, .
$$
The {\em non-signaling value} of a game $\game$ is defined as $\val_\ns(\game) \assign \sup_\q \val[\q](\game)$, where the supremum is over all non-signaling strategies $\q$ for $\game$. 
\end{definition}



The following relaxed notion of non-signaling is crucial for the understanding of our parallel-repetition proof. 

\begin{definition}\label{def:almostNS}
A {\em strategy} $\q$ for an $m$-player game $\game = (\X,\A,\pi,\V)$
is {\em $\eps$-almost non-signaling} if for any index subset $I
\subset \set{1,\ldots,m}$ and its complement $J = \set{1,\ldots,m}
\setminus I$, it holds that
$$
\Biggl|\sum_{a_J \in \A_J} \!\! \q(a_I,a_J|x_I,x_J) - \!\! \sum_{a_J \in \A_J} \!\! \q(a_I,a_J|x_I,x'_J) \Biggr| \leq \eps
\quad\text{for all $a_I \in \A_I$, $x_I \in \X_I$ and $x_J,x'_J \in \X_J$} \, .
$$
\end{definition}

\section{A Multi-Player Parallel Repetition Theorem}

\subsection{The Parallel Repetition of Nonlocal Games}

Given a game $\game$, the $n$-fold parallel repetition $\game^n$ is the game
where the referees samples $n$ independent inputs
$\allx=(x\ind{1},\ldots,x\ind{n}) \in \X \x \cdots \x \X$ and $\game^n$ is won if and only if all
its sub-games are won.
%
%
For the sake of notational convenience, we also introduce the
following way of denoting the fact that $t$ of the $n$ parallel
repetitions are won.
\begin{definition}[$t$-out-of-$n$ Parallel Repetition]
For any $n \in \N$ and $t \in \R$, the \emph{$t$-out-of-$n$ parallel repetition} of a game  $\game =
(\X,\A,\pi,\V)$ is given by the game $\game^{t/n} = (\X^n,\A^n,\pi^n,\V^{t/n})$
where $\X^n = \X \x \cdots \x \X$ and $\A^n = \A \x \cdots \x \A$, and for all $\allx=(x\ind{1},\ldots,x\ind{n}) \in \X^n$ and $\alla=(a\ind{1},\ldots,a\ind{n}) \in \A^n$
\begin{align*}
\pi^n(\allx) \assign \prod_{\ell=1}^n \pi(x\indl) 
\qquad\text{and}\qquad
\V^{t/n}(\allx,\alla) \assign 
\left\{\begin{array}{ll}
1 & \text{if } \sum_{\ell=1}^n \V(x\indl,a\indl) \geq t \\
0 & \text{else }
\end{array}\right. \, .
\end{align*}
The (standard) \emph{$n$-fold parallel repetition} of a game $\game$ is given by the game $\game^n := \game^{n/n}$. 
\end{definition}

Similar to the observation after Definition~\ref{def:game}, for any game $\game$ and for any strategy\footnote{We write $\qn$ (rather than e.g.~$\q^n$) to emphasize that it is a strategy for an $n$-fold repetition of $\game$, but it is {\em not} (necessarily) the $n$-fold independent execution of a strategy $\q$ for $\game$. } 
$\qn$ for the $t$-out-of-$n$ (or the $n$-fold) parallel repetition, random variables $\allX = (X\ind{1},\ldots,X\ind{n})$ and $\allA = (A\ind{1},\ldots,A\ind{n})$, together with their joint distribution $P_{\allX \allA}$, are naturally determined. 

Note that for any $\ell \in \set{1,\ldots,n}$, $X\indl$ is of the form $X\indl = (X_1\indl,\ldots,X_m\indl)$, where $X_i\indl$ represents the question to the $i$-th player in the $\ell$-th repetition of $\game$ (and is distributed over $\X_i$). Therefore, for any $i \in \set{1,\ldots,m}$, we write $\allX_i$ for $\allX_i = (X_i\ind{1},\ldots,X_i\ind{n})$, and for any $I = \set{i_1,\ldots,i_k} \subseteq \set{1,\ldots,m}$, $X_I\indl$ should be understood as $X_I\indl = (X_{i_1}\indl,\ldots,X_{i_k}\indl)$ and  $\allX_I$ as  $\allX_I = (\allX_{i_1},\ldots,\allX_{i_k})$. 
The corresponding holds for $\allA$.

To simplify notation, for the $n$-fold repetition of a given game $\game$ with a given strategy $\qn$,  we define $W_\ell$ to be the random variable $W_\ell \assign \V(X\indl,A\indl)$ that indicates if the $\ell$-th repetition of $\game$ is won, and we define $\av{W} \assign \frac{1}{n}\sum_{\ell=1}^n W_\ell$ to be the fraction of repetitions that are won. Obviously, $\val[\qn](\game^{t/n}) = \prob{\av{W} \geq t/n}$.

\subsection{Concentration and Parallel Repetition Theorems}

Our concentration and parallel repetition theorems below hold for all multi-player nonlocal games $\game$ up to the following restriction on the distribution $\pi$. 
\begin{definition}
We say that an $m$-player game $\game = (\X,\A,\pi,\V)$ has {\em complete support} if $\pi(x) > 0$ for all $x \in \X$, i.e., every $x \in \X = \X_1\x\cdots\x\X_m$ is a ``valid input'' to the game.  
\end{definition}
An important class of games that satisfy the complete-support property are the so-called {\em free} games, as studied for instance in~~\cite{BRRRS09}. In a free game, $\pi$ is required to be a {\em product distribution}, i.e., $\pi(x) = \pi_1(x_1) \cdots \pi_m(x_m)$ for all $x = (x_1,\ldots,x_m) \in \X = \X_1 \times \ldots \times \X_m$. Such a game has obviously full support.\footnote{After possibly having restricted the sets $\X_1,\ldots,\X_m$ appropriately. }

\begin{theorem}[Concentration Theorem]\label{thm:CT}
Let $\game$ be an arbitrary $m$-player game with complete
support. 
Then there exists a constant $\mu > 0$, depending on $\game$, such that for any $\delta > 0$, any $n \in \N$, and for $t = (\val_\ns(\game) \!+\! \delta)n$: 
$$
\val_\ns(\game^{t/n}) \leq 8 \exp\bigl(-\delta^4 \mu n \bigr) \; .
$$
\end{theorem}
As an immediate consequence, we get the following parallel-repetition theorem. 

\begin{theorem}[Parallel-Repetition Theorem]\label{thm:PR}
Let $\game$ be an arbitrary $m$-player game with complete
support and non-signaling value $\val_\ns(\game) < 1$. Then there exists $\nu < 1$, depending on $\game$, such that  $\val_\ns(\game^n) < 8 \nu^n$ for any $n \in \N$. 
\end{theorem}
We point out that the constants $\mu$ (in Theorem~\ref{thm:CT}) and $\nu$ (in Theorem~\ref{thm:PR})  not only depend on the non-signaling value $\val_\ns(\game)$ of $\game$, but on the game $\game$ itself. The restriction to games with complete support stems from the fact that $\mu$ becomes $0$ when the smallest probability in the distribution $\pi$ goes to~$0$, rendering the bound useless.

\subsection{The Proof}

A central idea of our proof is the {\em robustness} of the
non-signaling value of a game.
We will use the following result from~\cite[Section~10.4]{Schrijver98} about the
sensitivity analysis of linear programs. 
\begin{lemma} \label{lem:sensitivity}
Let $A$ be an $m \times n$-matrix, and let $A$ be such that for each
nonsingular submatrix $B$ of $A$, all entries of $B^{-1}$ are at most
$\Delta$ in absolute value. Let $c$ be a row $n$-vector, and let $b'$ and $b''$ be column $m$-vectors such that
both $\max_x \Set{cx}{Ax \leq b'}$ and $\max_x \Set{cx}{Ax \leq b''}$
are finite. Then \footnote{For $x=(x_1,\ldots,x_n)
  \in \R^n$, the norms are defined as $\|x\|_1 = \sum_i |x_i|$ and
  $\|x\|_\infty = \max_i |x_i|$.}
\begin{align*}
\left| \, \max_{x \in \R^n} \Set{cx}{Ax\leq b''} - \max_{x \in \R^n} \Set{cx}{Ax \leq b'} \, \right|
\leq n \Delta \|c\|_1 \cdot \|b'' - b'\|_\infty \,  .
\end{align*}
\end{lemma}

\begin{prop}[Robustness of $\val_\ns(\game)$]\label{prop:robust}
  Let $\game$ be an $m$-player game with non-signaling value
  $\val_\ns(\game)$. Then, there exists a constant $\const(\game)$
  such that for any $\eps \geq 0$ and for any strategy $\q$ for
  $\game$ that is $\eps$-almost non-signaling, the value of $\game$
  with respect to $\q$ is bounded by $\val[\q](\game) \leq
  \val_\ns(\game) + \const(\game) \cdot \eps$.
\end{prop}
\begin{proof}
The non-signaling value $\val_\ns(\game)$ is the optimal value of the
following linear program:
\begin{align}
&\mbox{maximize}\quad\sum_{x \in \X \atop a\in \A} \pi(x) \, \V(x,a)
    \, q(a|x) \nonumber \\
&\mbox{subject to}\\
&\quad q(a|x) \geq 0 \quad\text{for all $a \in \A$, $x
  \in \X$,} \label{cons:positive}\\
&\quad \sum_{a \in \A} q(a|x) = 1 \quad\text{for all $x \in \X$,} \label{cons:sum}\\
\begin{split} \label{cons:nonsignaling}
&\quad \!\!\! \sum_{a_J \in \A_J} \!\! \q(a_I,a_J|x_I,x_J) -
\q(a_I,a_J|x_I,x'_J) = 0 \quad \text{for all $I \subset \set{1,\ldots,m}$,
  $J=\set{1,\ldots,m} \setminus I$}\\[-3mm]
&\qquad\qquad\qquad\qquad\qquad\qquad\qquad\qquad\qquad\text{ and for all $a_I \in \A_I$, $x_I \in \X_I$ and $x_J,x'_J \in \X_J$}
\, .
\end{split}
\end{align}
Lemma~\ref{lem:sensitivity} gives a bound on how much the optimal
value of this linear program can vary if we optimize over
$\eps$-almost non-signaling strategies instead of a fully
non-signaling strategies. Formally, we can express the linear program
above in the ``standard form'' $\max\Set{cx}{Ax \leq b'}$ by expanding
the equality constraints \eqref{cons:sum} and
\eqref{cons:nonsignaling} as $\leq$ and $\geq$ inequality
constraints. According to Definition~\ref{def:almostNS},
$\eps$-almost non-signaling strategies fulfill the
constraints~\eqref{cons:nonsignaling} only up to an error of
at most $2\eps$. Hence, relaxing the constraints from non-signaling to
$\eps$-almost non-signaling amounts to change the $b'$-coordinates
corresponding to the non-signaling constraints
\eqref{cons:nonsignaling} from $0$ to $2\eps$. Hence, the parameters of
Lemma~\ref{lem:sensitivity} are $\| b'' - b' \|_\infty = 2\eps$,
$n=|\X| \cdot |\A|$, $\|c\|_1 = \sum_{x \in \X \atop a\in \A} |\pi(x)
\, \V(x,a) | \leq |\A|$ and $\Delta$ is a finite constant that depends
on the number of players $m$ and the number of answers $|\A|$ and
questions $|\X|$.\footnote{In our case, the relevant constraint matrix $A$ has $n = |\X| \cdot |\A|$ columns and at most $2\left( (|\A|\cdot|\X| + |\X|^2)^m + |\X|\right)$ rows. Let $\Delta := \max \left\{ \left|(B^{-1})_{ij}\right| \mid \mbox{B a nonsingular submatrix of $A$} \right\}$, which depends only $m, |\A|, |\X|$.
}
 Finally, we note that we can apply the lemma,
because the objective function is at most one (and thus finite)
irrespective of which strategies we are considering. Setting
$\const(\game) \assign 2|\X| |\A|^2 \Delta$ yields the claim. 
\end{proof}

\begin{lemma}[Main Lemma]\label{lemma:main}
Let $\game$ be a game with complete support. 
Consider an $n$-fold repetition $\game^n$ of $\game$ with an arbitrary non-signaling strategy $\qn$ for $\game^n$. Let $\ev$ be an arbitrary event (in the underlying probability space). 
Then for any subset $S = \set{v_1,\ldots,v_k} \subset \set{1,\ldots,n}$, the probability $\Prob{W_{V}\!=\!1}{\ev}$ for a randomly chosen $V$ in $\set{1,\ldots,n} \setminus S$ is bounded by
$$
\Prob[\big]{W_{V}\!=\!1}{\ev}\leq \val_\ns(\game) + \const'(\game) \cdot \sqrt{\textstyle\frac{1}{n-k} \log\bigl(\frac{1}{\prob{\ev}}\bigr)}
$$
where $\const'(\game) = 3\cdot 2^m \const(\game)/\min_x\pi(x)$ is some constant that only depends on $\game$. 
\end{lemma}
The following is an immediate consequence. 

\begin{corollary}\label{cor:main}
  Let $\game$ be a game with complete support.  Consider an execution
  of the $n$-fold repetition $\game^n$ with an arbitrary non-signaling
  strategy for $\game^n$. For any $\ell \in \set{1,\ldots,n}$, let
  $\ev_\ell$ be the event that the $\ell$-th
  repetition is accepted, i.e.~$W_\ell = 1$.  Then for any subset $S =
  \set{v_1,\ldots,v_k} \subset \set{1,\ldots,n}$, there exists $v_{k+1}
  \in \set{1,\ldots,n} \setminus S$ such that
$$
\Prob[\big]{\ev_{v_{k+1}}}{\ev_{v_1} \wedge \ldots\wedge \ev_{v_k}} \leq \val_\ns(\game) + \const'(\game) \cdot \sqrt{\textstyle\frac{1}{n-k} \log\bigl(\frac{1}{\prob{\ev_{v_1} \wedge \ldots\wedge \ev_{v_k}}}\bigr)}
$$
where $\const'(\game)$ is some constant that only depends on $\game$.
\end{corollary}

\begin{proof}[Proof (of Lemma~\ref{lemma:main})] 
Let $\pio > 0$ be such that $\pi(x) \geq \pio$ for all $x \in \X$;
by assumption on $\game$, such a $\pio$ exists. 
By re-ordering the (strategies of the) $n$ executions, we may assume without loss of generality that $S = \set{n-k+1,\ldots,n}$, and we now need to argue about the probability over a random $V$ in $\set{1,\ldots,n-k}$. To simplify notation, let us define
$$
\eps:= \sqrt{\textstyle\frac{1}{n-k} \log\bigl(\frac{1}{\prob{\ev}}\bigr)} \, .
$$
Fix a subset $I \subseteq \set{1,\ldots,m}$ and let $J = \set{1,\ldots,m} \setminus I$ be the complement of $I$. 
Consider the distribution 
$$
P_{\allX_I \allX_J \allA_I} = P_{\allX_I \allA_I} \cdot P_{\allX_J|\allX_I \allA_I} 
= P_{\allX_I \allA_I} \cdot P_{\allX_J|\allX_I} 
= P_{\allX_I \allA_I} \cdot \prod_{\ell=1}^n P_{X_J\indl|\allX_I} 
= P_{\allX_I \allA_I} \cdot \prod_{\ell=1}^n P_{X_J\indl|\allX_I \allA_I} 
$$ 
where the second equality is due to non-signaling, the third due to the independence of every pair $(X_I\indl,X_J\indl)$, and the third again due to non-signaling. We can thus apply Lemma~\ref{lemma:bound} (with $T = (\allX_I,\allA_I)$ and $U\indl = X_J\indl$) and obtain 
\begin{align*}
(n-k) &\cdot \eps = \sqrt{\textstyle(n-k) \log\bigl(\frac{1}{\prob{\ev}}\bigr)} \, \geq \sum_{\ell=1}^{n-k} \bigl\| P_{\allX_I X_J\indl \allA_I|\ev} - P_{\allX_I \allA_I | \ev} \cdot P_{X_J\indl|\allX_I\allA_I} \bigr\|  \\
&\geq  \sum_{\ell=1}^{n-k} \bigl\| P_{X_I\indl X_J\indl A_I\indl|\ev} - P_{X_I\indl A_I\indl | \ev} \cdot P_{X_J\indl|X_I\indl A_I\indl} \bigr\| 
 = \sum_{\ell=1}^{n-k} \bigl\| P_{X_I\indl X_J\indl A_I\indl|\ev} - P_{X_I\indl A_I\indl | \ev} \cdot P_{X_J\indl|X_I\indl} \bigr\| \\
&= \sum_{\ell=1}^{n-k} \bigl\| P_{X_I\indl X_J\indl|\ev} \cdot P_{A_I\indl| X_I\indl X_J\indl \ev} - P_{X_I\indl| \ev} \cdot P_{A_I\indl| X_I\indl \ev} \cdot P_{X_J\indl|X_I\indl} \bigr\| \, .
\end{align*}
The first inequality holds by Lemma~\ref{lemma:bound}. The second
inequality follows from Lemma~\ref{lemma:marginals} which states that
the distance of the random variables $X_I\indl, X_J\indl, A_I\indl$ cannot be
larger than the distance of all random variables $\allX_I, X_J\indl, \allA_I$.
The subsequent equality holds due to the non-signaling condition between subsets $I$ and $J$, and the last equality is a simple re-writing of some probabilities. 

By means of Lemma~\ref{lemma:bound} (setting $T$ to be a constant), we can also conclude that $\sum_\ell \|P_{X_I\indl X_J\indl|\ev} - P_{X_I\indl X_J\indl} \|$, and thus in particular $\sum_\ell \|P_{X_I\indl|\ev} - P_{X_I\indl} \|$, is upper bounded by $(n-k)\eps$. Therefore, noting that $P_{X_I\indl X_J\indl} = P_{X_I X_J}$, we can conclude that
$$
\sum_{\ell=1}^{n-k} \bigl\| P_{X_I X_J} \cdot P_{A_I\indl| X_I\indl X_J\indl \ev} - P_{X_I X_J} \cdot P_{A_I\indl| X_I\indl \ev} \bigr\| 
\leq 3 (n-k) \eps \, .
$$

By summing over all subsets $I \subseteq \set{1,\ldots,m}$ (and letting $J$ be its complement), changing the order of the summation, and defining 
$$
\eps_\ell \assign \sum_I \bigl\| P_{X_I X_J} \cdot P_{A_I\indl| X_I\indl X_J\indl \ev} - P_{X_I X_J} \cdot P_{A_I\indl| X_I\indl \ev} \bigr\|
$$ 
we get
$$
\sum_{\ell=1}^{n-k} \eps_\ell
\leq 3\cdot 2^m (n-k) \eps\, .
$$
Note that by definition of $\eps_\ell$, for any choice of $I$ and $J = \set{1,\ldots,m} \setminus I$, it holds that 
$$
\bigl\| P_{X_I X_J} \cdot P_{A_I\indl| X_I\indl X_J\indl \ev} - P_{X_I X_J} \cdot P_{A_I\indl| X_I\indl \ev} \bigr\| 
\leq \eps_\ell \, ,
$$
and hence, by the lower bound $\pio$ on $P_{X_I X_J}$, that 
$$
\bigl\| P_{A_I\indl| X_I\indl X_J\indl \ev}(\cdot|x_I,x_J) - P_{A_I\indl| X_I\indl \ev}(\cdot|x_I) \bigr\| 
\leq  \frac{\eps_\ell}{\pio} 
$$
for any $x_I$ and $x_J$.  For any $\ell \in \set{1,\ldots, n-k}$,
consider the strategy $\tilde{\q}_\ell$ for (one execution of)
$\game$, defined by $\tilde{\q}_\ell(a|x) = P_{A\indl| X\indl
  \ev}(a|x)$. By the above, $\tilde{\q}_\ell$ is
$(\eps_\ell/\pio)$-almost non-signaling.  
Furthermore, by the definition of $\tilde{\q}_\ell$, the
probability $\Prob{\ev_\ell}{\ev}$ that the
$\ell$-th repetition of the $n$-fold repetition of $\game$ is
accepted  equals the probability $\val[\tilde{\q}_\ell](\game)$ that a
{\em single} execution of $\game$ is accepted when strategy
$\tilde{\q}_\ell$ is played. Since $\tilde{\q}_\ell$ is
$(\eps_\ell/\pio)$-almost non-signaling, it follows from
Proposition~\ref{prop:robust} that this probability is at most
$\val_{\ns}(\game)+\const(\game)\cdot\eps_\ell/\pio$.  The claimed
bound on $\Prob{\ev_{V}}{\ev}$ for a randomly chosen $V$ in
$\set{1,\ldots,n-k}$ now follows from the bound on $\sum_\ell
\eps_\ell$, where $\const'(\game)$ is given by $3\cdot 2^m
\const(\game)/\pio$.
\end{proof}

We are now ready to prove our main concentration bound.
\begin{proof}[Proof (of Theorem~\ref{thm:CT})]
Let $K$ be some integer parameter, to be defined later. 
Let $V_1,\ldots,V_K$ be a random subset of distinct integers from $\set{1,\ldots,n}$, and let $D_k$ be the random variable $D_k = W_{V_k} =  \V(X\ind{V_k},A\ind{V_k})$ for any $k \in \set{1,\ldots,K}$. Understanding $V_1,\ldots,V_K$ as a ``sample subset" of the $n$ parallel repetitions of $\game$, $D_k$ indicates whether the $k$-th game in the sample is won.  
A pair $(d_1,\ldots,d_k) \in \set{0,1}^k$ and $(v_1,\ldots,v_k) \in \set{1,\ldots,n}$ of $k$-tuples is called {\em typical} if $P_{D_1\cdots D_k|V_1\cdots V_k}(d_1,\ldots, d_k|v_1,\ldots, v_k) \geq 2^{-2K}$. Let $\typ_k$ be the event that $(D_1\cdots D_k)$ and $(V_1\cdots V_k)$ form a typical pair.
Note that the corresponding complementary events satisfy $\bar{\typ}_k \Rightarrow \bar{\typ}_{k+1}$ as well as 
$$
\prob{\bar{\typ}_k} = \!\!\!\sum_{\text{atypical pairs} \atop (d_1...d_k),(v_1... v_k)}\!\!\! P_{V_1\cdots V_k}(v_1,\ldots, v_k) \, P_{D_1\cdots D_k|V_1\cdots V_k}(d_1\cdots d_k|v_1\cdots v_k) < 2^{-K} \, .
$$
Let $\gamma \assign 1 -  \val_\ns(\game) - \eps$ where $\eps \assign \delta/3$. Note that we obviously may assume that $\delta \leq 1-\val_\ns(\game)$ so that $\gamma > 0$.
We now define a sequence of random variables $M_0,\ldots,M_K$ as follows. Random variable $M_0$ takes the value $0$ with certainty, and $M_{k+1}$  is inductively defined as
$$
M_{k+1} \assign \left\{\begin{array}{ll}
M_{k} + \gamma & \text{if $D_{k+1} \!=\! 1 \mbox{ and } \typ_{k}$ }\\
M_{k} - (1-\gamma) & \text{otherwise} \, .
\end{array}\right.
$$
We want to show that $M_0,\ldots,M_K$ forms a supermartingale.  We fix
$k \in \set{0,\ldots,K-1}$ and we fix values $(v_1,\ldots,v_{k})$ for
the random variables $V_1,\ldots,V_{k}$. Up to the end of this
paragraph, all probabilities etc. are to be understood conditioned on
these values. We define $\ev$ to be the event that $D_1,\ldots,D_{k}$
take on some arbitrary but fixed values $(d_1,\ldots,d_{k})$.  If the
pair $(d_1,\ldots,d_{k})$ and $(v_1,\ldots,v_{k})$ is atypical, then
conditioned on $\ev$ we have $M_{k+1} = M_{k} + \gamma -1 < M_{k}$ and
thus $\E[M_{k+1}|M_0 \cdots M_{k}] < \E[M_{k}|M_0 \cdots M_{k}] =
M_{k}$.  In the other case, if the pair $(d_1,\ldots,d_{k})$ and
$(v_1,\ldots,v_{k})$ is typical then $\prob{\ev} \geq
2^{-2K}$. Furthermore, Lemma~\ref{lemma:main} implies that
$P_{D_{k+1}|\ev}(1) = \Prob{\ev_{V_{k+1}}}{\ev} \leq \val_\ns(\game) +
\const'(\game) \sqrt{\log(1/\prob{\ev})/(n-k)} \leq \val_\ns(\game) +
\const'(\game) \sqrt{2K/(n-K)}$.  We want this last term to be upper bounded by
$\val_\ns(\game) + \eps = 1-\gamma$, which we achieve by choosing $K$
as $K \assign \lfloor \alpha n \rfloor$ where $\alpha \assign
\min\set{\eps^2/(3\const'(\game)^2), 1/3}$, as can easily be verified. 
It follows that $\E[M_{k+1}|M_0
\cdots M_{k}]  \leq (1-\gamma)(M_{k}+\gamma)+\gamma(M_{k} - (1-\gamma)) = M_{k}$ (when conditioning on~$\ev$). Since the argument that the $M_0, \ldots, M_K$ form a supermartingale holds independent of the choice of $(d_1,\ldots,d_{k})$ and of the
choice of $(v_1,\ldots,v_{k})$, $M_0,\ldots,M_K$ indeed forms a
supermartingale in the original probability space (without
conditioning on the values for $V_1,\ldots,V_{k}$).  Therefore,
\begin{align*}
\prob[\bigg]{\sum_{k=1}^K &D_k \geq (\val_\ns(\game)\!+\!2\eps) K} 
\,\leq\, \prob[\big]{\bar{\typ}_K} + \prob[\big]{M_K \geq (\val_\ns(\game)\!+\!2\eps)K\gamma - (1\!-\!\val_\ns(\game)\!-\!2\eps)K(1\!-\!\gamma)} \\
&\leq\, 2^{-K} + \prob[\big]{M_K \geq  (\gamma-1+\val_\ns(\game)+2\eps)K} 
\,=\, 2^{-K} + \prob{M_K \!\geq\! \eps K}  \\[1ex]
&\leq\, 2^{-K} + \exp(-\eps^2K/2)
\,<\, 2\exp(-\eps^2K/2)
\end{align*}
The first inequality holds by definition of $M_K$, and the second by a simple manipulation of the terms. The equality holds by definition of $\gamma$, and the subsequent inequality by the Azuma-Hoeffding Inequality. Finally, the last inequality holds since $\eps < 1$ and $\exp(\frac12) < 2$. 

On the other hand, setting $\av{D} \assign \frac{1}{K}\sum_{k=1}^K D_k$, we can also write
\begin{align*}
\prob[\big]{\av{D} \geq \val_\ns(\game)+2\eps}  
\geq\prob[\big]{\av{W}  > \val_\ns(\game) +\delta} \cdot \Prob[\big]{\av{D} \geq \val_\ns(\game)+2\eps}{\av{W}  > \val_\ns(\game) +\delta}
\end{align*}
where by the Hoeffding Inequality (and using that $\eps = \delta/3$)
$$
\Prob[\big]{\av{D} \geq \val_\ns(\game)\!+\!2\eps}{\bar{W}  > \val_\ns(\game) +\delta} \geq 1 - \exp(-2\eps^2 K) \, .
$$
Therefore,
$$
\prob[\big]{\av{W}  > \val_\ns(\game) +\delta} 
\leq  \frac{2\exp(-\eps^2K/2)}{1-\exp(-2\eps^2K)} \, .
$$
In case that $\exp(-2\eps^2K) < \frac14$, we obtain the bound
\begin{equation} \label{eq:lastbound}
\prob[\big]{\av{W}  > \val_\ns(\game) +\delta} 
\leq  \frac{8}{3} \exp(-\eps^2K/2) \, .
\end{equation}
Note that in the other case, if $\exp(-2\eps^2K) \geq \frac14$, then
$2\exp(-\eps^2K/2) \geq 1$ and the bound~\eqref{eq:lastbound} holds trivially.

Setting $\mu \assign 1/(2\cdot 3^5 \cdot \const'(\game)^2)$, and recalling that $\eps = \delta/3$ and $K \assign \lfloor \alpha n \rfloor$ with $\alpha$ chosen as $\alpha
\assign \min\set{\eps^2/(3\const'(\game)^2), 1/3}$, leads to the
claim. 
\end{proof}

\section{Conclusion and Open Questions}
This article initiates the investigation of the behavior of multi-player nonlocal games under parallel repetition. For the case of the non-signaling value, we provide a concentration bound for games with complete support. Our results might serve as a stepping stone for the investigation of the quantum and classical values. Other interesting questions include improving the rate of repetition (e.g.~by making it independent of the minimal probability that any question is asked) or finding cryptographic applications, for instance in position-based cryptography.

\section*{Acknowledgments}
We would like to thank Tobias M\"uller for insightful discussions about the
sensitivity of linear programs.
HB is supported by a 7th framework EU SIQS grant. CS is supported by an NWO VENI grant.

\bibliographystyle{plain}

\bibliography{qip,crypto,procs}

\begin{thebibliography}{10}

\bibitem{BRRRS09}
Boaz Barak, Anup Rao, Ran Raz, Ricky Rosen, and Ronen Shaltiel.
\newblock Strong parallel repetition theorem for free projection games.
\newblock In Irit Dinur, Klaus Jansen, Joseph Naor, and Jos{\'e} D.~P. Rolim,
  editors, {\em APPROX-RANDOM}, volume 5687 of {\em Lecture Notes in Computer
  Science}, pages 352--365. Springer, 2009.

\bibitem{BBLV13}
Jop Bri\"{e}t, Harry Buhrman, Troy Lee, and Thomas Vidick.
\newblock Multipartite entanglement in xor games.
\newblock {\em Quantum Information \& Computation}, 13(3-4):334--360, March
  2013.

\bibitem{BCFGGOS11}
Harry Buhrman, Nishanth Chandran, Serge Fehr, Ran Gelles, Vipul Goyal, Rafail
  Ostrovsky, and Christian Schaffner.
\newblock Position-based quantum cryptography: Impossibility and constructions.
\newblock In Phillip Rogaway, editor, {\em Advances in Cryptology – CRYPTO
  2011}, volume 6841 of {\em Lecture Notes in Computer Science}, pages
  429--446. Springer Berlin / Heidelberg, 2011.

\bibitem{BFSS13}
Harry Buhrman, Serge Fehr, Christian Schaffner, and Florian Speelman.
\newblock The garden-hose model.
\newblock In {\em Innovations in Theoretical Computer Science, ITCS '13,
  Berkeley, CA, USA, January 9-12, 2013}, pages 145--158. ACM, 2013.

\bibitem{CS13arxiv}
A.~{Chailloux} and G.~{Scarpa}.
\newblock {Parallel Repetition of Entangled Games with Exponential Decay via
  the Superposed Information Cost}.
\newblock arxiv:1310.7787, 2013.

\bibitem{CSUU08}
Richard Cleve, William Slofstra, Falk Unger, and Sarvagya Upadhyay.
\newblock Perfect parallel repetition theorem for quantum xor proof systems.
\newblock {\em Computational Complexity}, 17(2):282--299, 2008.

\bibitem{DSV13arxiv}
I.~{Dinur}, D.~{Steurer}, and T.~{Vidick}.
\newblock {A parallel repetition theorem for entangled projection games}.
\newblock arxiv:1310.4113, 2013.

\bibitem{hoeffding63}
Wassily Hoeffding.
\newblock Probability inequalities for sums of bounded random variables.
\newblock {\em Journal of the American Statistical Association},
  58(301):13--30, March 1963.

\bibitem{Holenstein09}
Thomas Holenstein.
\newblock Parallel repetition: Simplification and the no-signaling case.
\newblock {\em Theory of Computing}, 5(1):141--172, 2009.

\bibitem{JPY13arxiv}
R.~{Jain}, A.~{Pereszl{\'e}nyi}, and P.~{Yao}.
\newblock {A parallel repetition theorem for entangled two-player one-round
  games under product distributions}.
\newblock arxiv:1311.6309, 2013.

\bibitem{KRT10}
Julia Kempe, Oded Regev, and Ben Toner.
\newblock Unique games with entangled provers are easy.
\newblock {\em SIAM J. Comput.}, 39(7):3207--3229, July 2010.

\bibitem{KV11}
Julia Kempe and Thomas Vidick.
\newblock Parallel repetition of entangled games.
\newblock In {\em Proceedings of the 43rd annual ACM symposium on Theory of
  computing}, STOC '11, pages 353--362, New York, NY, USA, 2011. ACM.

\bibitem{Rao11}
Anup Rao.
\newblock Parallel repetition in projection games and a concentration bound.
\newblock {\em SIAM J. Comput.}, 40(6):1871--1891, 2011.

\bibitem{Raz98}
Ran Raz.
\newblock A parallel repetition theorem.
\newblock {\em SIAM J. Comput.}, 27(3):763--803, June 1998.

\bibitem{Rosen10}
Ricky Rosen.
\newblock A k-provers parallel repetition theorem for a version of no-signaling
  model.
\newblock {\em Discrete Math., Alg. and Appl.}, 2(4):457--468, 2010.

\bibitem{Schrijver98}
A.~Schrijver.
\newblock {\em Theory of Linear and Integer Programming}.
\newblock Wiley Series in Discrete Mathematics \& Optimization. John Wiley \&
  Sons, 1998.

\bibitem{TFKW13}
Marco Tomamichel, Serge Fehr, Jędrzej Kaniewski, and Stephanie Wehner.
\newblock A monogamy-of-entanglement game with applications to
  device-independent quantum cryptography.
\newblock {\em New Journal of Physics}, 15(10):103002, 2013.

\end{thebibliography}

\end{document}